\documentclass{llncs}

\usepackage[british]{babel}

\usepackage{graphicx}

\usepackage{tikz}
\usetikzlibrary{arrows,positioning,patterns,decorations.pathreplacing}

\usepackage{hyperref}
\usepackage[capitalise,nameinlink]{cleveref}
\crefname{section}{Sect.}{Sect.}
\Crefname{section}{Section}{Sections}
\crefname{figure}{Fig.}{Fig.}
\Crefname{figure}{Figure}{Figures}

\usepackage{xspace}
\newcommand{\eg}{e.\,g.}

\usepackage{amsmath,mathtools}

\usepackage{syntax} 
\usepackage{stmaryrd} 

\usepackage{algorithm}
\usepackage{algpseudocode}

\newcommand{\cpplang}{\textsc{C++}\xspace}

\newcommand{\terms}{\mathcal{T}\xspace}
\newcommand{\booleans}{\mathcal{B}\xspace}

\newcommand{\rel}{\sqsubseteq}
\newcommand{\constr}[2]{\ensuremath{#1 \rel #2}}
\newcommand{\upvar}[1]{#1^\uparrow}
\newcommand{\downvar}[1]{#1^\downarrow}
\newcommand{\var}[1]{#1}

\newcommand{\utvars}{\ensuremath{\mathsf{V}^{\uparrow}}}
\newcommand{\dtvars}{\ensuremath{\mathsf{V}^{\downarrow}}}

\newcommand{\bvars}{\ensuremath{\mathsf{V}^b}}

\newcommand{\nil}{\ensuremath{\mathsf{nil}}}
\newcommand{\none}{\ensuremath{\mathsf{none}}}
\newcommand{\true}{\ensuremath{\mathsf{true}}}
\newcommand{\false}{\ensuremath{\mathsf{false}}}

\newcommand{\cspws}{\textsc{CSP-WS}}

\newcommand{\con}{\ensuremath{\mathcal{C}}}
\newcommand{\ssat}{\ensuremath{\mathsf{B}}}
\newcommand{\unsat}{\ensuremath{\mathsf{Unsat}}}

\newcommand{\wf}{\ensuremath{\textsf{WFC}}}
\newcommand{\sen}{\ensuremath{\textsf{SC}}}

\newcommand{\af}{\ensuremath{\textsf{AF}}}

\newcommand{\tbot}{\underline{\bot}}
\newcommand{\ttop}{\overline{\top}}

\usepackage{etoolbox}

\newcommand{\el}[3]{%
    \mathsf{#1}\ifx$#2$\else(#2)\fi{:}~{#3}
}

\newcounter{listcount}\newcounter{totalcount}%
\newcommand{\tuple}[1]{%
  \setcounter{totalcount}{0}
  \renewcommand*{\do}[1]{\stepcounter{totalcount}}
  \docsvlist{#1}
  \setcounter{listcount}{0}
  \renewcommand*{\do}[1]{
    \stepcounter{listcount}
    ##1\ifnum\value{listcount}<\value{totalcount}\,\fi
  }
  \ensuremath{(\docsvlist{#1})}
}
\newcommand{\tlist}[2]{%
  \setcounter{totalcount}{0}
  \renewcommand*{\do}[1]{\stepcounter{totalcount}}
  \docsvlist{#1}
  \setcounter{listcount}{0}
  \renewcommand*{\do}[1]{
    \stepcounter{listcount}
    ##1\ifnum\value{listcount}<\value{totalcount},\else\fi
  }
  \ensuremath{[\docsvlist{#1}\ifx$#2$\else\,| #2\fi]}
}
\newcommand{\record}[2]{%
  \setcounter{totalcount}{0}
  \renewcommand*{\do}[1]{\stepcounter{totalcount}}
  \docsvlist{#1}
  \setcounter{listcount}{0}
  \renewcommand*{\do}[1]{
    \stepcounter{listcount}
    ##1\ifnum\value{listcount}<\value{totalcount}, \else\fi
  }
  \ensuremath{\{\docsvlist{#1}\ifx$#2$\else\,| #2\fi\}}
}
\newcommand{\choice}[2]{%
  \setcounter{totalcount}{0}
  \renewcommand*{\do}[1]{\stepcounter{totalcount}}
  \docsvlist{#1}
  \setcounter{listcount}{0}
  \renewcommand*{\do}[1]{
    \stepcounter{listcount}
    ##1\ifnum\value{listcount}<\value{totalcount},\else\fi
  }
  \ensuremath{{(}{:}\docsvlist{#1}\ifx$#2$\else\,| #2\fi{:}{)}}
}
\newcommand{\switch}[1]{%
  \setcounter{totalcount}{0}
  \renewcommand*{\do}[1]{\stepcounter{totalcount}}
  \docsvlist{#1}
  \setcounter{listcount}{0}
  \renewcommand*{\do}[1]{
    \stepcounter{listcount}
    ##1\ifnum\value{listcount}<\value{totalcount},\else\fi
  }
  \ensuremath{\langle\docsvlist{#1}\rangle}
}

\newcommand{\eel}[3]{%
    \mathsf{#1}\ifx$#2$\else(#2)\fi{:}~{#3}
}

\newcommand{\etuple}[1]{%
  \setcounter{totalcount}{0}
  \renewcommand*{\do}[1]{\stepcounter{totalcount}}
  \docsvlist{#1}
  \setcounter{listcount}{0}
  \renewcommand*{\do}[1]{
    \stepcounter{listcount}
    ##1\ifnum\value{listcount}<\value{totalcount}\,\fi
  }
  \ensuremath{(\docsvlist{#1})}
}
\newcommand{\etlist}[2]{%
  \setcounter{totalcount}{0}
  \renewcommand*{\do}[1]{\stepcounter{totalcount}}
  \docsvlist{#1}
  \setcounter{listcount}{0}
  \renewcommand*{\do}[1]{
    \stepcounter{listcount}
    ##1\ifnum\value{listcount}<\value{totalcount},\else\fi
  }
  \ensuremath{[\docsvlist{#1}\ifx$#2$\else\,| #2\fi]}
}
\newcommand{\erecord}[2]{%
  \setcounter{totalcount}{0}
  \renewcommand*{\do}[1]{\stepcounter{totalcount}}
  \docsvlist{#1}
  \setcounter{listcount}{0}
  \renewcommand*{\do}[1]{
    \stepcounter{listcount}
    ##1\ifnum\value{listcount}<\value{totalcount}, \else\fi
  }
  \ensuremath{\{\docsvlist{#1}\ifx$#2$\else\,| #2\fi\}}
}
\newcommand{\echoice}[2]{%
  \setcounter{totalcount}{0}
  \renewcommand*{\do}[1]{\stepcounter{totalcount}}
  \docsvlist{#1}
  \setcounter{listcount}{0}
  \renewcommand*{\do}[1]{
    \stepcounter{listcount}
    ##1\ifnum\value{listcount}<\value{totalcount},\else\fi
  }
  \ensuremath{{(}{:}\docsvlist{#1}\ifx$#2$\else\,| #2\fi{:}{)}}
}
\newcommand{\eswitch}[1]{%
  \setcounter{totalcount}{0}
  \renewcommand*{\do}[1]{\stepcounter{totalcount}}
  \docsvlist{#1}
  \setcounter{listcount}{0}
  \renewcommand*{\do}[1]{
    \stepcounter{listcount}
    ##1\ifnum\value{listcount}<\value{totalcount},\else\fi
  }
  \ensuremath{\langle\docsvlist{#1}\rangle}
}

\newcommand{\eupvar}[1]{\textsf{#1}^\uparrow}
\newcommand{\edownvar}[1]{\textsf{#1}^\downarrow}

\begin{document}

\title{A Constraint Satisfaction Method for Configuring Non-Local Service
Interfaces}

\author{Pavel Zaichenkov \and Olga Tveretina \and Alex Shafarenko}
\institute{%
    University of Hertfordshire, United Kingdom
}

\maketitle

\begin{abstract}
Modularity and decontextualisation are core principles of a service-oriented
architecture.  However, the principles are often lost when it comes to an
implementation of services, as a result of a rigidly defined service interface.
The interface, which defines a data format, is typically specific to
a particular context and its change entails significant redevelopment costs.
This paper focuses on a two-fold problem.  On the one hand, the interface
description language must be flexible enough for maintaining service
compatibility in a variety of different contexts without modification of the
service itself.  On the other hand, the composition of interfaces in
a distributed environment must be provably consistent.  The existing approaches
for checking compatibility of service choreographies are either inflexible
(WS-CDL and WSCI) or require behaviour specification associated with each
service, which is often impossible to provide in practice.

We present a novel approach for automatic interface configuration in
distributed stream-connected components operating as closed-source services
(i.e.\ the behavioural protocol is unknown).  We introduce a Message Definition
Language (MDL), which can extend the existing interfaces description languages,
such as WSDL, with support of subtyping, inheritance and polymorphism.  The MDL
supports configuration variables that link input and output interfaces of
a service and propagate requirements over an application graph.  We present an
algorithm that solves the interface reconciliation problem using constraint
satisfaction that relies on Boolean satisfiability as a subproblem.
\end{abstract}


\section{Introduction}

For the last decade service-oriented computing (SOC) has been a promising
technology facilitating development of large scale distributed systems.  SOC
allows enterprises to expose their internal business systems as services
available on the Internet.  On the other hand, clients can combine services and
reuse them for developing their own applications or constructing more complex
services.  Although web services continue to play an important role in modern
software development, a service composition is still a key challenge for SOC
and web services.  Web service composition empowers organisations to build
inter-enterprise software, to outsource software modules, and to provide an
easily accessible functionality for their customers.  Furthermore, service
composition reduces the cost and risks of new software development, because the
software elements that are represented as web services can be reused
repeatedly~\cite{sheng2014web}.

Web Service Description Language (WSDL) is an XML-based specification language
for describing service interfaces, which is a \emph{de facto} standard in
SOC\@.  Functionality and compatible data formats of the service are specified
in WSDL in the form of an interface.  The names and formats in the interfaces
of communicating services must exactly match for interface compatibility.
Today the environment in which services are developed and executed has become
more open, changing and dynamic, which requires  an adaptable and flexible
approach to service composition.  The choreography wired to specific WSDL
interfaces is too restrictive for dependable service composition.  The
choreography is statically bounded to specific operation names and types, which
impedes reusability of compound services and their interaction descriptions.

Reliable and dependable service composition remains a significant challenge
today~\cite{sheng2014web,dustdar2005survey,sheng2014behavior}.  Services are
provided autonomously by various organisations.  Developers of applications,
particularly safety-critical applications, such as health care, stock trading,
nuclear systems, must be able to check soundness and completeness of service
composition at early stages.  Therefore, model checking and verification of web
services is being actively researched
today~\cite{sheng2014behavior,bourne2012ensuring,zheng2013personalized}.

Web Services Choreography Description Language (WS-CDL)~\cite{web2002web} and
Web Service Choreography Interface (WSCI)~\cite{arkin2002web} are languages for
describing protocols from a global perspective.  This approach is based on
$\pi$-calculus that defines a behavioural semantics for concurrent processes.
An application designer writes a global description in WS-CDL or WSCI that
should be realisable by local protocols of communicating services.  Service
interfaces in WS-CDL are specified in WSDL\@.  The relation between service
interfaces connected with a communication channel is one-to-one, i.e.\ there is
no way to propagate data format requirements and capabilities across the
communication graph if services are not explicitly connected by a channel.
Moreover, \cite{barros2005critical} emphasises that the existing association
between WS-CDL and WSDL does not allow equivalent services with different WSDL
interfaces to be part of the choreography.

Session types is another approach based on $\pi$-calculus that assures
communication safety in distributed systems and in service choreographies
particularly~\cite{carbone2007structured}.  A choreography is defined as
a global protocol in terms of the interactions that are expected from the
protocol peers and a set of local protocols, one for each peer, which describes
the global protocol from the viewpoint of an individual peer.  The session
types require services to expose their behaviour as a protocol.  This
information is enough to define a communication type system, which is
well-suited for verifying runtime properties of the system such as
deadlock-freedom, interleaving, etc.  The session types essentially rely on
behavioural protocols, which in most cases are neither explicitly provided nor
can be derived from the code.

In this paper we present a formal method for configuring flexible interfaces
based on constraint satisfaction and SAT\@.  In contrast to the approaches
based on $\pi$-calculus, our method does not require services to define
a protocol, but only to specify the data interface.


\section{Motivating Example}
\label{sec:example}

Our approach for configuring web services is motivated by rapid development of
Cloud computing, social networks and Internet of Things, which accelerate the
growth and complexity of service
choreographies~\cite{bouguettaya2013advanced,duan2012survey,mathew2013web}.
Accordingly, we chose a simple but non-trivial example from one of those areas
to illustrate our approach.  The same example, known as the {\em three-buyer
use case}, is  often called upon to demonstrate the capabilities of session
types such as communication safety, progress and session fidelity
guarantees~\cite{coppo2015gentle,honda2008multiparty}.

\begin{figure}[t]
    \centering
    \begin{tikzpicture}[thick,node distance=2cm,minimum width=3em]
        \node[draw,rectangle] (a) {Alice};
        \node[draw,rectangle,right=of a] (s) {Seller};
        \node[draw,rectangle,right=of s] (b) {Bob};
        \node[draw,rectangle,right=of b] (c) {Carol};

        \draw[transform canvas={yshift=0.5ex},->] (a) to node[above,near start] {$AS_{\text{out}}$} node[above,near end] {$AS_{\text{in}}$} (s);
        \draw[transform canvas={yshift=-0.5ex},->,dashed] (s) to node[below,near start] {$SA_{\text{out}}$} node[below,near end] {$SA_{\text{in}}$} (a);

        \draw[transform canvas={yshift=0.5ex},->] (s) to node[above,near start] {$SB_{\text{out}}$} node[above,near end] {$SB_{\text{in}}$} (b);
        \draw[transform canvas={yshift=-0.5ex},->,dashed] (b) to node[below,near start] {$BS_{\text{out}}$} node[below,near end] {$BS_{\text{in}}$} (s);

        \draw[transform canvas={yshift=0.5ex},->] (b) to node[above,near start] {$BC_{\text{out}}$} node[above,near end] {$BC_{\text{in}}$} (c);
        \draw[transform canvas={yshift=-0.5ex},->,dashed] (c) to node[below,near start] {$CB_{\text{out}}$} node[below,near end] {$CB_{\text{in}}$} (b);
    \end{tikzpicture}
    \caption{Service composition in a Three Buyer usecase}
\label{fig:three-buyer}
\end{figure}
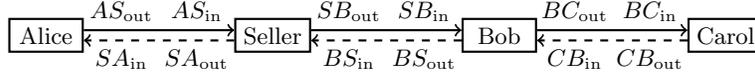

Consider a system involving buyers called Alice, Bob and Carol that cooperate
in order to buy a book from a Seller.  Each buyer is specified as an
independent service that is connected with other services via a channel-based
communication.  There is an interface associated with every input and output
port of a service, which specifies the service's functionality and data formats
that the service is compatible with.  The interfaces are defined in a Message
Definition Language (MDL) that is formally introduced in \cref{sec:mdl}.
\cref{fig:three-buyer} depicts composition of the application where Alice is
connected to Seller only and can interact with Bob and Carol indirectly.  $AS,
SB, BC, CB, BS, AS$ denote interfaces that are associated with service
input/output ports.  For brevity, we only provide $AS, SB$ and $BC$ (the rest
of the interfaces are defined in the same manner), which are specified in the
MDL as terms in the following way:
\begin{equation*}
    \scriptsize
    \openup-.5\jot
\begin{aligned}[c]
    AS_{\text{out}} = {(}{:}&\eel{request}{}{\erecord{\eel{title}{}{\edownvar{tv}}}{}},\\
                     &\eel{payment}{}{\erecord{\eel{title}{}{\edownvar{tv}},
                      \eel{money}{}{\mathsf{int}},
                      \eel{id}{}{\mathsf{int}}}{}},\\
                     &\eel{share}{x}{\erecord{\eel{title}{}{\edownvar{tv}},
                      \eel{money}{}{\mathsf{int}}}{}},\\
                             &\eel{suggest}{y}{\erecord{\eel{title}{}{\edownvar{tv}}}{}}{:}{)}\\
                      SB_{\text{out}} = {(}{:}&\mathsf{response}:~\{\eel{title}{}{\mathsf{string}},\eel{money}{}{\mathsf{int}}\}\\
                                              &|~\eupvar{ct1}{:}{)}\\
    BC_{\text{out}} = {(}{:}&\eel{share}{z}{\erecord{\eel{quote}{}{\mathsf{string}}, \eel{money}{}{\mathsf{int}}}{}}~|~\eupvar{ct2}{:}{)}
\end{aligned}
\left|
\begin{aligned}[c]
    AS_{\text{in}} = {(}{:}&\eel{request}{}{\erecord{\eel{title}{}{\mathsf{string}}}{}},\\
                            &\eel{payment}{}{\erecord{\eel{title}{}{\mathsf{string}},
                            \eel{money}{}{\mathsf{int}}}{}}\\
                            &|~\eupvar{ct1}{:}{)}\\
                            SB_{\text{in}} = {(}{:}&\mathsf{share}(z):~\{\eel{quote}{}{\mathsf{string}},\\
                                                   &\eel{money}{}{\mathsf{\mathsf{int}}}\},\\
                           &\mathsf{response}:~\{\eel{title}{}{\mathsf{string}},\eel{money}{}{\mathsf{int}}\}\\
                           &|~\eupvar{ct2}{:}{)}\\
    BC_{\text{in}} = {(}{:}&\mathsf{share}:~\{\eel{quote}{}{\mathsf{string}}, \eel{money}{}{\mathsf{\mathsf{int}}}\}{:}{)}
\end{aligned}
\right.
\end{equation*}

${(}{:}~{:}{)}$ delimit a collection of alternative label-record pairs called
\emph{variants}, where the label corresponds to the particular implementation
that can process a message defined by the given record.  A \emph{record}
delimited by ${\{}~{\}}$ is a collection of label-value pairs.  Collection
elements may contain Boolean variables called \emph{guards} (\eg{} $x$, $y$ or
$z$ in our example).  A guard instantiated to $\false$ excludes the
element from the collection.  This
is the main self-configuration mechanism: Boolean variables control the
dependencies between any elements of interface collections (this can be seen as
a generalised version of intersection types~\cite{davies2000intersection}) The
variables exclude elements from the collection if the dependencies between
corresponding elements in the interfaces that are connected by a communication
channel cannot be satisfied.

Parametric polymorphism is supported using interface variables, such as
$\edownvar{tv}$, $\eupvar{ct1}$ and $\eupvar{ct2}$ (the meaning of
$\uparrow$ and $\downarrow$ is explain in \cref{sec:mdl}).  Moreover, the
presence of $\eupvar{ct1}$ and $\eupvar{ct2}$ in both input and output
interfaces enables flow inheritance~\cite{grelck2008gentle} mechanism that
provides delegation of the data and service functionality across available
services.

$AS_\text{out}$ declares an output interface of Alice, which declares
functionality and a format of messages sent to Seller.  The service has the
following functionality:

\begin{itemize}
    \item Alice can \textsf{request} a book's price from Seller by providing
        a \textsf{title} of an arbitrary type (which is specified by a term
        variable $\edownvar{tv}$) that Seller is compatible with.  On the other
        hand, Seller declares that a title of type \textsf{string} is only
        acceptable, which means that $\edownvar{tv}$ must be instantiated to
        \textsf{string}.
    \item Furthermore, Alice can provide a \textsf{payment} for a book.  In
        addition to the \textsf{title} and the required amount of
        \textsf{money}, Alice provides her \textsf{id} in the message.
        Although Seller does not require the \textsf{id}, the interconnection
        is still valid (a description in standard WSDL interfaces would cause
        an error though) due to the subtyping supported in the MDL\@.
    \item Furthermore, Alice can offer to \textsf{share} a purchase between
        other customers.  Although Alice is not connected to Bob or Carol and
        may even not be aware of their presence (the example illustrates
        a composition where some service communicates with services that the
        service is not directly connected with), our mechanism detects that Alice
        can send a message with ``\textsf{share}'' label to Bob by bypassing it
        implicitly through Seller.  In order to enable inheritance in Seller's
        service, the mechanism sets a tail variable $\eupvar{ct1}$ to ${(}{:}
        \eel{share}{}{\{\eel{title}{}{\textsf{string}},
        \eel{money}{}{\textsf{int}}\}}{:}{)}$.  If Bob were unable to
        accept a message with ``\textsf{share}'' label, the mechanism would
        instantiate $x$ with $\false$, which automatically removes the
        corresponding functionality from the service.
    \item Finally, Alice can \textsf{suggest} a book to other buyers.  However,
        examination of other service interfaces shows that there is no service
        that can receive a message with the label ``\textsf{suggest}''.
        Therefore, a communication error occurs if Alice decides to send the
        message.  To avoid this, the configuration mechanism excludes
        ``\textsf{suggest}'' functionality from Alice's service by setting
        $y$ variable to $\false$.
\end{itemize}

The proposed configuration mechanism analyses the interfaces of services
Seller, Bob and Carol in the same manner.  The presence of $\eupvar{ct1}$
variable in both input and output interfaces of Bob enables support of data
inheritance on the interface level.  Furthermore, the Boolean variable $z$
behaves as an intersection type: Bob has ``purchase sharing'' functionality
declared as an element $\el{share}{z}{\{\dots\}}$ in its input interface
$SB_\text{in}$ (used by Seller).  The element is related to the element
$\el{share}{z}{\{\dots\}}$ in its output interface $BC_\text{out}$ (used by
Carol).  The relation declares that Bob provides Carol with ``sharing''
functionality only if Bob was provided with the same functionality from Seller.
In our example, $z$ is $\true$, because Carol declares that it can receive
messages with the label ``\textsf{share}''.  Note that there could be an any
Boolean formula in place of $z$, which wires any input and output interfaces of
a single service in an arbitrary way.  The existing interface description
languages (WSDL, WS-CDL, etc.) do not support such interface wiring
capabilities.

Interface variables provide facilities similar to \cpplang{} templates.
Services can specify a generic behaviour compatible with multiple contexts and
input/output data formats.  Given the context, the compiler then specialises
the interfaces based on the requirements and capabilities of other services.

The problem being solved is similar to type inference problem; however, it has
large combinatorial complexity and, therefore, direct search of a solution is
impractical.  Furthermore, additional complexity arises from the presence of
Boolean variables in general form.  Another problem is potential cyclic
dependencies in the network, which prevent the application of a simple forward
algorithm.  In our approach, we define our problem as a constraint satisfaction
problem.  Then we employ a constraint solver, which was specifically developed
to solve this problem, to find correct instantiations of the variables.

\section{Message Definition Language and CSP}
\label{sec:mdl}

Now we define a term algebra called Message Definition Language (MDL).  The
purpose of the MDL is to describe flexible service interfaces.  Although we use
a concise syntax for MDL terms that is different from what standard WSDL-based
interfaces look like, it can easily be rewritten as a WSDL extension.

In our approach, a message is a collection of data entities, each specified by
a corresponding \emph{term}.
The intention of the term is to represent
\begin{enumerate}
\item a standard atomic type such as \texttt{int}, \texttt{string}, etc.;
\item inextensible data collections such as tuples;
\item extensible data
    records~\cite{gaster1996polymorphic,leijen2005extensible}, where
    additional named fields can be introduced without breaking the match
    between the producer and the consumer and where fields can also be
    inherited from input to output records by lowering the output type, which
    is always safe;
\item  data-record variants, where generally more variants can be accepted by
    the consumer than the producer is aware of, and where such additional
    variants can be inherited from the output back to the input of the producer
    --- hence contravariance --- again, by raising the input type, which is
    always safe, too.
\end{enumerate}

\subsection{Terms}

Each term is either atomic or a collection in its own right.  Atomic terms are
\emph{symbols}, which are identifiers used to represent standard types such as
\texttt{int}, \texttt{string}, etc.  To account for subtyping we include three
categories of collections: \emph{tuples} that are demanded to be of the same
size and thus admit only depth structural subtyping, \emph{records} that are
subtyped covariantly (a larger record is a subtype) and \emph{choices} that are
subtyped contravariantly using set inclusion (a smaller choice is a subtype).

In order to support parametric polymorphism and inheritance in interfaces, we
introduce term variables (called later \emph{t-variables}), which are similar
to type variables. For coercion of interfaces it is important to distinguish
between two variable categories: \emph{down-coerced} and \emph{up-coerced}
ones. The former can be instantiated with symbols, tuples and records (terms of
these three categories are call down-coerced terms), and the latter can only be
instantiated with choices (up-coerced terms). Informally, for two down-coerced
terms, a term associated with a structure with ``more data'' is a subtype  of
the one associated with a structure that contains less; and vice versa for
up-coerced terms.  We use the notation $\downvar{v}$ and $\upvar{v}$ for
down-coerced and up-coerced variables respectively, and $v$ when its coercion
sort is unimportant. Explicit sort annotation on variables is useful for
simplifying partial order definitions on terms.



We introduce Boolean variables (called  \emph{b-variables} below) in the term
interfaces to specify dependencies between input and output data formats.
B-variables provide functionality similar to intersection types, which increase
the expressiveness of function signatures.

A Boolean expression $b \in \booleans$ ($\booleans$ denotes a set of Boolean
expressions) called a guard is defined by the following grammar: \vspace{-.4em}
\begin{grammar}
    <guard> ::= (<guard> $\land$ <guard>) | (<guard> $\lor$ <guard>)
               | <guard> $\to$ <guard> | $\lnot$<guard> | "true" | "false" | $b$-$variable$
\end{grammar}
\vspace{-.4em}

MDL terms are built recursively using the constructors: tuple, record, choice
and switch, according to the following grammar:
\vspace{-.4em}
\begin{grammar}
    <term> ::=  <symbol> | <tuple> | <record> | <choice> | t-variable

    <tuple> ::= "("<term> $[$<term>$]^*$")"

    <record> ::= "\{"$[$<element>$[$","<element>$]^* [$"|"down-coerced t-variable$]]$"\}"

    <choice> ::= "(:"$[$<element>$[$","<element>$]^*[$"|"up-coerced t-variable$]]$":)"

	<element> ::= <label>"("<guard>"):"<term>
	
    <label> ::= <symbol>

\end{grammar}
\vspace{-.4em}

Informally, a \emph{tuple} is an ordered collection of terms and
a \emph{record} is an extensible, unordered collection of guarded labeled
terms, where \emph{labels} are arbitrary symbols, which are unique within
a single record.  A \emph{choice} is a collection of alternative terms.  The
syntax of choice is the same as that of record except for the delimiters.  The
difference between records and choices is in width subtyping and will become
clear below when we define seniority on terms.  We use choices to represent
polymorphic messages and service interfaces on the top level.  Records and
choices are defined in \emph{tail form}.  The tail is denoted by a t-variable
that represents a term of the same kind as the construct in which it occurs.

A switch is an auxiliary construct intended for building conditional terms,
which is specified as a set of unlabelled (by contrast to a choice) guarded
alternatives.  Formally, it is defined as
\vspace{-.4em}
\begin{grammar}
<switch> ::= "<"<guard>":"<term>$[$", "<guard>":"<term>$]^*$">"
\end{grammar}
\vspace{-.4em}

Exactly one guard must be \true{} for any valid switch, i.e. the switch is
substitutionally equivalent to the term marked by the \true{} guard:
\[
    \switch{\el{}{\false}{t_1},\dots,\el{}{\true}{t_i},\dots,\el{}{\false}{t_n}} = \switch{\el{}{\true}{t_i}} = t_i\text{.}
\]

For example, $\switch{\el{}{a}{int}, \el{}{\lnot a}{string}}$ represents the
symbol $int$ if $a = \true$, and the symbol $string$ otherwise.

\subsection{Seniority Relation}

For a guard $g$, we denote as  $\bvars(g)$  the set of b-variables that
occur in $g$.  For a term $t$, we denote as $\utvars(t)$ the set of
up-coerced t-variables that occur in $t$,  and as $\dtvars(t)$ the set of
down-coerced ones; and finally $\bvars(t)$ is  the set of b-variables in $t$.

\begin{definition}[Semi-ground and ground terms]
    A term $t$ is called semi-ground if $\utvars(t) \cup \dtvars(t)
    = \emptyset$.  A term $t$ is called ground if it is semi-ground and
    $\bvars(t) = \emptyset$.
\end{definition}

\begin{definition}[Well-formed terms]
\label{defn:well-formed-term}
    A term $t$ is well-formed if it is ground and exactly one of the following
    holds:
    \vspace{-.8em}
    \begin{enumerate}
        \item $t$ is a symbol;
        \item $t$ is a tuple $\tuple{t_1 \dots t_n}$, $n > 0$, where all
            $t_i$, $1 \leq i \leq n$, are well-formed;
         \item $t$ is a record
             $\record{\el{l_1}{g_1}{t_1},\dots,\el{l_n}{g_n}{t_n}}{}$ or
             a choice $\choice{\el{l_1}{g_1}{t_1},\dots,\el{l_n}{g_n}{t_n}}{}$,  $n \geq 0$,
             where for all $1 \leq i\ne j \leq n$, $g_i\wedge g_j \implies l_i
             \neq l_j$ and all $t_i$ for which $g_i$ are \true{} are
             well-formed;
        \item $t$ is a switch $\switch{\el{}{g_1}{t_1},
            \dots,\el{}{g_n}{t_n}}$, $n>0$, where for some  $1\leq  i\leq n$,
            $g_i=\true$ and $t_i$ is well-formed and where $g_j=\false$ for all
            $j \neq i$.
    \end{enumerate}
\end{definition}

If an element of a record, choice or switch has a guard that is equal to
\false, then the element can be omitted, \eg{}
\[
    \record{\el{a}{x \land y}{string}, \el{b}{\false}{int},\eel{c}{x}{int}}{} = \record{\el{a}{x \land y}{string}, \eel{c}{x}{int}}{}\text{.}
\]
If an element of a record or a choice has a guard that is \true, the guard can
be syntactically omitted, \eg{}
\[
    \record{\el{a}{x \land y}{string}, \el{b}{\true}{int},\el{c}{x}{int}}{} = \record{\el{a}{x \land y}{string}, \el{b}{}{int},\el{c}{x}{int}}{}\text{.}
\]
We define the \emph{canonical form} of a well-formed collection as
a representation that does not include \false{} guards, and we omit \true{}
guards anyway.  The canonical form of a switch is its (only) term with
a \true{} guard, hence any term in canonical form is switch-free.

Next we introduce a seniority relation on terms for the purpose of structural
subtyping.  In the sequel we use $\nil$ to denote the empty record
$\{\enskip\}$, which has the meaning of unit type and represents a message
without any data.  Similarly, we use $\none$ to denote the empty choice
$(:\enskip:)$.

\begin{definition}[Seniority relation]
\label{def:seniority}
    The seniority relation $\rel$ on well-formed terms is defined in canonical
    form as follows:
    \vspace{-.8em}
    \begin{enumerate}
        \item $\none \rel t$ if $t$ is a choice;
        \item $t \rel \nil$ if $t$ is a symbol, a tuple or a record;
        \item $t \rel t$;
        \item $t_1 \rel t_2$, if for some $k,m>0$ one of the following holds:
        \begin{enumerate}
            \item $t_1 = \tuple{t^1_1,\dots,t^k_1}$, $t_2
                = \tuple{t^1_2,\dots,t^k_2}$ and $t^i_1 \rel t^i_2$ for each
                $1\leq i \leq k$;

            \item $t_1
                = \record{\el{l^1_1}{}{t^1_1},\dots,\el{l^k_1}{}{t^k_1}}{}$ and
                $t_2
                = \record{\el{l^1_2}{}{t^1_2},\dots,\el{l^m_2}{}{t^m_2}}{}$,
                where $k \geq m$ and for each $j \leq m$ there is $ i \leq k$
                such that $l^i_1 = l^j_2$ and $t^i_1 \rel t^j_2$;

            \item $t_1
                = \choice{\el{l^1_1}{}{t^1_1},\dots,\el{l^k_1}{}{t^k_1}}{}$ and
                $t_2
                = \choice{\el{l^1_2}{}{t^1_2},\dots,\el{l^m_2}{}{t^m_2}}{}$,
                where $k \leq m$ and for each $i \leq k$ there is $ j \leq m$
                such that $l^i_1 = l^j_2$ and $t^i_1 \rel t^j_2$.
        \end{enumerate}
    \end{enumerate}
\end{definition}

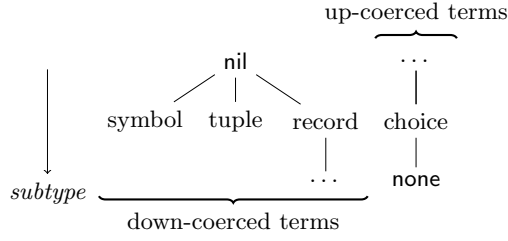
\begin{figure}[t]
    \begin{center}
        \begin{tikzpicture}[node distance=.8cm,line width=.1mm]
    \node (nil) at (0,0) {\nil};
    \node [below of=nil] (tuple) {tuple};

    \node [right of=tuple,node distance=1.2cm] (record) {record};
    \node [below of=record] (recordEnd) {$\dots$};

    \node [left of=tuple,node distance=1.2cm] (symbol) {symbol};

    \node [right of=record,node distance=1.2cm] (choice) {choice};
    \node [above of=choice] (choiceEnd) {$\dots$};
    \node [below of=choice] (none) {\none};

    \node (dummy1) at (-2.5, 0) {};
    \node [below] (dummy2) at (-2.5, -1.5) {{\em subtype}};

    \draw (nil) -- (symbol);
    \draw (nil) -- (tuple);
    \draw (nil) -- (record);
    \draw (choice) -- (none);

    \draw (record) -- (recordEnd);
    \draw (choice) -- (choiceEnd);

    \draw (choice) -- (choiceEnd);

    \draw[ultra thin,->] (dummy1) -- (dummy2);

\draw [
    thick,
    decoration={
        brace,
        mirror,
        raise=1cm
    },
    decorate
] (symbol.west) -- (record.east) 
node [pos=0.5,anchor=north,yshift=-1.1cm] {down-coerced terms}; 

\draw [
    thick,
    decoration={
        brace,
        raise=1.1cm
    },
    decorate
] (choice.west) -- (choice.east) 
node [pos=0.5,anchor=north,yshift=1.7cm] {up-coerced terms}; 

\end{tikzpicture}
    \end{center}
    \caption{Two semilattices representing the seniority relation for terms of
    different categories.  The lower terms are the subtypes of the upper ones}
\label{fig:semilattice}
\end{figure}

Similarly to the t-variables, terms are classified into two categories:
symbols, tuples and records are down-coerced terms and choices are up-coerced
terms.  The seniority relation defines a symmetric relation on down-coerced and
up-coerced terms: an element $\nil$ is the maximum element for down-coerced
terms; on the other hand, $\none$ is the minimum element for up-coerced terms.
$\downvar{\terms}$ denotes the set of all down-coerced ground terms,
$\upvar{\terms}$ denotes the set of all up-coerced ground terms and $\terms
= \downvar{\terms} \cup \upvar{\terms}$ is the set of all ground terms.
Similarly, $\downvar{\terms}_m$ denotes the set of all vectors of down-coerced
ground terms of length $m$ and $\upvar{\terms}_n$ denotes the set of all
vectors of up-coerced ground terms of length $n$.  If $\vec{t_1}$ and
$\vec{t_2}$ are vectors of terms $(t^1_1, \dots, t^1_n)$ and $(t^2_1, \dots,
t^2_n)$ of size $n$, then $\vec{t_1} \rel \vec{t_2}$ denotes the seniority
relation for all pairs $t^1_i \rel t^2_i$ ($1 \leq i \leq n$).

\begin{proposition}
\label{cor:semilattice}
    The seniority relation $\rel$ is a partial order, and $(\terms,
    \rel)$ is a pair of meet and join semilattices (\cref{fig:semilattice}):
    \begin{align*}
        \forall t_1, t_2 \in \downvar{\terms}, t_1 \rel t_2 &\text{ iff } t_1 \sqcap t_2 = t_1;\\
        \forall t_1, t_2\in \upvar{\terms}, t_1 \rel t_2 &\text{ iff } t_1 \sqcup t_2 = t_2.
    \end{align*}
\end{proposition}

The seniority relation represents the subtyping relation on terms.  If
a term $t$ describes the input interface of a service, then the service
can process any message described by a term $t'$, such that $t' \rel t$.

Although the seniority relation is straightforwardly defined for ground terms,
terms that are present in the interfaces of services can contain t-variables
and b-variables.  Finding such ground term values for the t-variables and such
Boolean values for the b-variables that the seniority relation holds represents
a CSP problem, which is formally introduced next.

\subsection{Constraint Satisfaction Problem for Web Services}
\label{sec:csp}

We define a substitution, which is used in the definition of the CSP and in the
algorithm, as a syntactic transformation that replaces b-variables with Boolean
values and t-variables with ground or semi-ground values.
\begin{definition}[Substitution]
    \label{def:substitution}
    Let $g$ be a guard, $t$ be a term, $k = |\bvars(g) \cup \bvars(t)|$, and
    $\vec{f} = (f_1,\dots,f_k)$ be a vector of b-variables contained in $g$ and
    $t$, and $\var{\vec{v}}=(\var{v_1},\dots,\var{v_k})$ be a vector of term
    variables contained in $t$.  Then for any vector of Boolean values $\vec{b}
    = (b_1,\dots,b_k)$ and a vector of terms $\vec{s} = (s_1,\dots,s_k)$
    \vspace{-.7em}
    \begin{enumerate}
        \item $g[\vec{f}/\vec{b}]$ denotes a Boolean value ($\true$ or
            $\false$), which is obtained as a result of the simultaneous
            replacement and evaluation of $f_i$ with $b_i$ for each $1 \leq
            i \leq k$;
        \item $t[\vec{f}/\vec{b}]$ denotes the vector obtained as a result of
            the simultaneous replacement of $f_i$ with $b_i$ for each
            $1\leq i \leq k$;
        \item $t[\var{\vec{v}}/\vec{s}]$ denotes the vector obtained as
            a result of the simultaneous replacement of $\var{v}_i$ with
            $s_i$ for each $1\leq i \leq k$;
        \item $t[\vec{f}/\vec{b},\var{\vec{v}}/\vec{s}]$ is a shortcut
            for $t[\vec{f}/\vec{b}][\var{\vec{v}}/\vec{s}]$.
    \end{enumerate}
\end{definition}

Given the set of constraints $\con$, we define the set of b-variables as
\[
    \bvars(\con) = \bigcup_{\constr{t}{t'} \in \con}\bvars(t) \cup \bvars(t')\text{,}
\]
the sets of of down-coerced and up-coerced t-variables as
\[
\dtvars(\con) = \bigcup_{\constr{t}{t'} \in \con}\dtvars(t) \cup \dtvars(t') \text{\quad and\quad} \utvars(\con) = \bigcup_{\constr{t}{t'} \in \con}\utvars(t) \cup \utvars(t')\text{.}
\]

In the following for each set of constraints $S$ such that $|\bvars(S)|=l$,
$|\utvars(S)|=m$ and $|\dtvars(S)|=n$ we use $\vec{f}=(f_1,\dots,f_l)$ to
denote the vector of b-variables contained in $S$,
$\upvar{\vec{v}}=(\upvar{v_1},\dots,\upvar{v_m})$ to denote the vector of
up-coerced t-variables and
$\downvar{\vec{v}}=(\downvar{v_1},\dots,\downvar{v_n})$ to denote the vector of
down-coerced t-variables.

Let $\con$ be a set of constraints such that $|\bvars(\con)|=l$,
$|\dtvars(\con)|=m$, $|\utvars(S)|=n$ and for some $l,m,n \geq 0$. Now we can
define a \cspws{} formally as follows.

\begin{definition}[\cspws{}]
\label{def:cspws}
    Find a vector of Boolean values $\vec{b} = (b_1,\dots,b_l)$ and vectors of
    ground terms $\downvar{\vec{t}}=(\downvar{t_1},\dots,\downvar{t_m})$,
    $\upvar{\vec{t}}=(\upvar{t_1},\dots,\upvar{t_n})$, such that for each
    $\constr{t_1}{t_2} \in \con$ 
    \[
        t_1[\vec{f}/\vec{b},\downvar{\vec{v}}/\downvar{\vec{t}},\upvar{\vec{v}}/\upvar{\vec{t}}] \rel
        t_2[\vec{f}/\vec{b},\downvar{\vec{v}}/\downvar{\vec{t}},\upvar{\vec{v}}/\upvar{\vec{t}}]
    \]
    The tuple $(\vec{b}, \downvar{\vec{t}}, \upvar{\vec{t}})$ is called a solution.
\end{definition}

\section{Solution Approximation}
\label{sec:approximation}

One way to solve \cspws{} is to attempt to solve the problem for all possible
instantiations of b-variables.  
We start with considering the simplification when the original problem is reduced to the one without
b-variables provided that some vector of Boolean assignments is given. 




We use an approximation algorithm that iteratively
traverses the meet and the join semilattices\footnote{The algorithm presented
here is widely used for data-flow analysis and flow graph
optimisations~\cite{kildall1973unified}.} for vectors of ground terms
$\downvar{\terms}_m$ and $\upvar{\terms}_n$, where $m = |\dtvars(\con)|$ and
$n = |\utvars(\con)|$, which represent solution approximations for down-coerced
and up-coerced terms respectively. The algorithm monotonically converges to
a solution if one exists.  Informally, the algorithm performs the following
steps:
\begin{enumerate}
    \item  Compute the initial approximation of the solution for $i=0$  as
        $(\downvar{\vec{a}}_i, \upvar{\vec{a}}_i) = ((\nil, \dots, \nil),
        (\none, \dots, \none))$, where the first element in the pair is the
        vector of top elements from the meet semilattice and the second element
        is the  vector of bottom elements from the join semilattice.
    \item Compute  $(\downvar{\vec{a}}_{i+1}, \upvar{\vec{a}}_{i+1})$ such that
        $\downvar{\vec{a}}_{i+1} \rel \downvar{\vec{a}}_{i}$ and
        $\upvar{\vec{a}}_{i} \rel \upvar{\vec{a}}_{i+1}$.
        \label{item:alg-loop}
    \item Repeat step~\ref{item:alg-loop} until a chain of approximations
        converges to the solution, i.e. $(\downvar{\vec{a}}_{i+1},
        \upvar{\vec{a}}_{i+1}) = (\downvar{\vec{a}}_{i}, \upvar{\vec{a}}_{i})$,
        or a situation where some of the constraints from
        step~\ref{item:alg-loop} cannot be satisfied.  Then return the last
        approximation as the solution or $\unsat$.
\end{enumerate}

We extend the set $\downvar{\terms}_m$ with the element $\tbot$, i.e.
$\downvar{\widetilde{\terms}}_m = \downvar{\terms}_m \cup \{\tbot\}$, and the
set $\upvar{\terms}_n$ with the element $\ttop$, i.e.
$\upvar{\widetilde{\terms}}_n = \upvar{\terms}_n \cup \{\ttop\}$.  Here $\tbot$
is defined as the bottom element of the meet semilattice, i.e. $\tbot \rel
\downvar{\vec{a}}$ for any  $\downvar{\vec{a}} \in
\downvar{\widetilde{\terms}}_m$, and  $\ttop$ is defined as the top element of
the join semilattice, i.e. $\upvar{\vec{a}} \rel \ttop$ for any
$\upvar{\vec{a}} \in \upvar{\widetilde{\terms}}_n$. The algorithm returns
$\tbot$ or $\ttop$ if it is unable to find an approximation for some
constraints, which, as shown in \cref{thm:cspkpn2-correct} below, means that
the input set of constraints does not have a solution.


\subsection{Approximating Function}
\label{sec:approximating-function}

In order to specify how the next approximation is computed we introduce the
\emph{approximating function} $\af: \con \times \downvar{\widetilde{\terms}}_m
\times \upvar{\widetilde{\terms}}_n \to \downvar{\widetilde{\terms}}_m \times
\upvar{\widetilde{\terms}}_n$ that maps a single constraint and the current
approximation to the new approximation.

The function $\af$ is given below for all categories of terms (except for
choices because they are symmetrical to the cases for records and switches that
are reduced to other term categories).  Let
$\downvar{\vec{v}} = (\overline{v}_1,\dots,\overline{v}_m)$, $\upvar{\vec{v}}
= (\overline{\overline{v}}_1,\dots,\overline{\overline{v}}_n)$,
$\downvar{\vec{a}} = (\overline{a}_1,\dots,\overline{a}_m)$,  $\upvar{\vec{a}}
= (\overline{\overline{a}}_1,\dots,\overline{\overline{a}}_n)$.

If $t$ is a symbol,  the given approximation $(\downvar{\vec{a}},
\upvar{\vec{a}})$ already satisfies the constraint:
\[
    \af(t \rel t, \downvar{\vec{a}}, \upvar{\vec{a}}) = (\downvar{\vec{a}}, \upvar{\vec{a}})\text{.}
\]

If $t$ is a down-coerced term  and $\overline{v}_l$ is a down-coerced variable,
the approximation for $\overline{v}_l$ is used to refine the approximation for
variables in $t$.  Therefore, the constraint is reduced to the one with
$\overline{v}_l$ as a ground term, which is obtained by substitution
$\overline{v}_l[\downvar{\vec{v}}/\downvar{\vec{a}}]$:
\[
    \af(t \rel \overline{v}_l, \downvar{\vec{a}}, \upvar{\vec{a}}) = \af(t \rel \overline{v}_l[\downvar{\vec{v}}/\downvar{\vec{a}}], \downvar{\vec{a}}, \upvar{\vec{a}})\text{.}
\]

If $\overline{\overline{{v}}}_l$ is an up-coerced variable and $t$ is an
up-coerced term, the case is symmetric to the previous one:
\[
    \af(\overline{\overline{{v}}}_l \rel t, \downvar{\vec{a}}, \upvar{\vec{a}}) = \af(\overline{\overline{{v}}}_l[\upvar{\vec{v}}/\upvar{\vec{a}}] \rel t, \downvar{\vec{a}}, \upvar{\vec{a}})\text{.}
\]

If $\overline{v}_l$ is a down-coerced variable and $t$ is a down-coerced term,
then $\overline{v}_l$ must be not higher than the ground term
$t[\downvar{\vec{v}}/\downvar{\vec{a}},\upvar{\vec{v}}/\upvar{\vec{a}}]$ in the
meet semilattice:
\[
    \af(\overline{v}_l \rel t, \downvar{\vec{a}}, \upvar{\vec{a}}) = ((\overline{a}_1,\dots,\overline{a}_l \sqcap t[\downvar{\vec{v}}/\downvar{\vec{a}},\upvar{\vec{v}}/\upvar{\vec{a}}],\dots,\overline{a}_m), \upvar{\vec{a}})\text{.}
\]

If $t$ is an up-coerced term and $\upvar{v}_l$ is an up-coerced variable, the
case is symmetric to the previous one:
\[
    \af(t \rel \upvar{v}_l, \downvar{\vec{a}}, \upvar{\vec{a}}) = (\downvar{\vec{a}}, (\overline{a}_1,\dots,\overline{a}_l \sqcup t[\downvar{\vec{v}}/\downvar{\vec{a}},\upvar{\vec{v}}/\upvar{\vec{a}}],\dots,\overline{a}_n))\text{.}
\]

If $t_1$ and $t_2$ are tuples $\tuple{t^1_1 \dots t^1_k}$ and $\tuple{t^2_1
\dots t^2_k}$ respectively, then the constraint must hold for the corresponding
nested terms:
\[
    \af(\tuple{t^1_1 \dots t^1_k} \rel \tuple{t^2_1 \dots t^2_k}, \downvar{\vec{a}}, \upvar{\vec{a}}) = (\bigsqcap_{1 \leq i \leq k} \downvar{\vec{a_i}}, \bigsqcup_{1 \leq i \leq k} \upvar{\vec{a_i}})\text{.}
\]

If $t_1$ and $t_2$ are records $\record{\el{l^1_1}{}{t^1_1}, \dots,
\el{l^1_p}{}{t^1_p}}{}$ and $\record{\el{l^2_1}{}{t^2_1}, \dots,
\el{l^2_q}{}{t^2_q}}{}$ respectively, two cases must be considered:
\begin{itemize}
    \item If for all $i$ ($1 \leq i \leq q$) there exists $j$ such
        that $l^1_j = l^2_i$, then the constraint for nested terms $t^1_j \rel
        t^2_i$ must hold:
        \[
            \af(\record{\el{l^1_1}{}{t^1_1}, \dots, \el{l^1_p}{}{t^1_p}}{} \rel \record{\el{l^2_1}{}{t^2_1}, \dots, \el{l^2_q}{}{t^2_q}}{}, \downvar{\vec{a}}, \upvar{\vec{a}}) = (\bigsqcap_{1 \leq i \leq q} \downvar{\vec{a_i}}, \bigsqcup_{1 \leq i \leq q} \upvar{\vec{a_i}})\text{.}
        \]
    \item Otherwise, the set of labels in $t_2$ is not a subset of the labels
        in $t_1$ and, therefore, $t_1 \rel t_2$ is unsatisfiable:
        \[
            \af(\record{\el{l^1_1}{}{t^1_1}, \dots, \el{l^1_p}{}{t^1_p}}{} \rel \record{\el{l^2_1}{}{t^2_1}, \dots, \el{l^2_q}{}{t^2_q}}{}, \downvar{\vec{a}}, \upvar{\vec{a}}) = (\tbot, \ttop)\text{.}
        \]
\end{itemize}

If $\overline{v}_l$ is a down-coerced variable, $t_1$ and $t_2$ are records
$\record{\el{l^1_1}{}{t^1_1}, \dots, \el{l^1_p}{}{t^1_p}}{\overline{v}_l}$ and
$\record{\el{l^2_1}{}{t^2_1}, \dots, \el{l^2_q}{}{t^2_q}}{}$ respectively, the
constraint can be satisfied only if for every nested term $t^2_i$ with the
label $l^2_j$ in $t$ one of the following holds: 1) there exists a subterm
$t^1_j$ with equal label in $t_1$ and $t^1_j \rel t^2_i$ holds, or 2)
$\overline{v}_l$ is a record that contains a junior to $t^2_i$ element with the
same label:
\[
    \af(\record{\el{l^1_1}{}{t^1_1}, \dots, \el{l^1_p}{}{t^1_p}}{\overline{v}_l} \rel \record{\el{l^2_1}{}{t^2_1}, \dots, \el{l^2_q}{}{t^2_q}}{}, \downvar{\vec{a}}, \upvar{\vec{a}}) = (\bigsqcap_{1 \leq i \leq q} \downvar{\vec{a_i}}, \bigsqcup_{1 \leq i \leq q} \upvar{\vec{a_i}})\text{,}
\]
where
\[
    (\downvar{\vec{a_i}}, \upvar{\vec{a_i}}) =
    \begin{cases}
        \af(t^1_j \rel t^2_i, \downvar{\vec{a_i}}, \upvar{\vec{a_i}})&\text{if } \exists j: l^1_j = l^2_i\\
        ((\overline{a}_1, \dots, \overline{a}_l \sqcap t^2_i[\downvar{\vec{v}}/\downvar{\vec{a}},\upvar{\vec{v}}/\upvar{\vec{a}}], \dots \overline{a}_m), \upvar{\vec{a}})&\text{otherwise.}
    \end{cases}
\]

If $t_1$ is a record $\record{\el{l^1_1}{}{t^1_1}, \dots,
\el{l^1_p}{}{t^1_p}}{}$ or $\record{\el{l^1_1}{}{t^1_1}, \dots,
\el{l^1_p}{}{t^1_p}}{\overline{v}_l}$ and $t_2$ is a record
$\record{\el{l^2_1}{}{t^2_1}, \dots, \el{l^2_q}{}{t^2_q}}{\overline{u}_r}$,
then the constraint can by substitution be reduced to the previous cases for
records:
 \[
    \af(t_1 \rel t_2, \downvar{\vec{a}}, \upvar{\vec{a}}) = \af(t_1 \rel t_2[\overline{u}_r/\overline{a}_r], \downvar{\vec{a}}, \upvar{\vec{a}})\text{.}
\]



The function $\af$ has the homomorphism property, which is important for
showing termination and correctness of the algorithm. 

\begin{lemma}[Homomorphism]
    \label{lem:homomoprhism}
    Let $\af(t_1 \rel t_2, \downvar{\vec{a}}_1, \upvar{\vec{a}}_1) = ({\downvar{\overline{\vec{a}}}}_1, {\upvar{\overline{\vec{a}}}}_1)$
    and $\af(t_1 \rel t_2, \downvar{\vec{a}}_2, \upvar{\vec{a}}_2) = ({\downvar{\overline{\vec{a}}}}_2, {\upvar{\overline{\vec{a}}}}_2)$.
    Then
    \[
        \af(t_1 \rel t_2, \downvar{\vec{a}}_1 \sqcap \downvar{\vec{a}}_2, \upvar{\vec{a}}_1 \sqcup \upvar{\vec{a}}_2) = ({\downvar{\overline{\vec{a}}}}_1 \sqcap {\downvar{\overline{\vec{a}}}}_2, {\upvar{\overline{\vec{a}}}}_1 \sqcup {\upvar{\overline{\vec{a}}}}_2).
    \]
\end{lemma}


%
%
The  function $\af_\con$ is a composition of $\af$ functions that are
sequentially applied to all constraints in $\con$ (the order in which $\af$ is
applied to the constraints is not important due to distributivity of the
semi-lattices):
\[
    \af_\con(\downvar{\vec{a}}, \upvar{\vec{a}}) = \af(t^{|\con|}_1 \rel t^{|\con|}_2, \af(t^{|\con|-1}_1 \rel t^{|\con|-1}_2, \dots, \af(t^1_1 \rel t^1_2, \downvar{\vec{a}}, \upvar{\vec{a}})\dots))\text{.}
\]
The sequential composition preserves homomorphism for $\af_\con$.  In
\cref{sec:algorithm} we tacitly assume that for arbitrary terms the function
$\af_\con$ is defined in a similar way.

\subsection{Fixed-Point Algorithm}

Now we present the algorithm (see \cref{alg:cspkpn2}) that computes a chain of
approximations for the case $\bvars(\con) = \emptyset$ that converges to the
solution if one exists.


\begin{algorithm}[htb]
    \caption{$\cspws(\con)$, where $\bvars(\con) = \emptyset$}
\label{alg:cspkpn2}
\begin{algorithmic}[1]
    \State$i \gets 0$
    \State$(\downvar{\vec{a}}_0, \upvar{\vec{a}}_0) \gets ((\nil,\dots,\nil), (\none,\dots,\none))$
    \Repeat
        \State$i \gets i + 1$
        \State$(\downvar{\vec{a}}_i, \upvar{\vec{a}}_i) \gets \af_\con(\downvar{\vec{a}}_{i-1}, \upvar{\vec{a}}_{i-1})$
    \Until{$(\downvar{\vec{a}}_i, \upvar{\vec{a}}_i) = (\downvar{\vec{a}}_{i-1}, \upvar{\vec{a}}_{i-1})$}
    \If{$(\downvar{\vec{a}}_i, \upvar{\vec{a}}_i) = (\tbot, \ttop)$}
    	\State\Return\unsat
    \Else
        \State\Return $(\downvar{\vec{a}}_i, \upvar{\vec{a}}_i)$
    \EndIf
   \end{algorithmic}
\end{algorithm}


\begin{theorem}[Termination]
    \label{thm:termination}
    For any set of constraints $\con$ such that $\bvars(\con) = \emptyset$,
    \cref{alg:cspkpn2} terminates after a finite number of steps.
\end{theorem}
\begin{proof}
    $\af_\con$ is a monotonic function that maps $\downvar{\vec{a}}_{i-1} \in
    \downvar{\widetilde{\terms}}_m$ to $\downvar{\vec{a}}_{i} \in \downvar{\widetilde{\terms}}_m$ and
    $\upvar{\vec{a}}_{i-1} \in \upvar{\widetilde{\terms}}_n$ to
    $\upvar{\vec{a}}_{i} \in \upvar{\widetilde{\terms}}_n$, where $\downvar{\vec{a}}_{i-1}$ and
    $\downvar{\vec{a}}_{i}$ are elements of the lattice
    $(\downvar{\widetilde{\terms}}_m, \rel)$ such that $\downvar{\vec{a}}_{i}
    \rel \downvar{\vec{a}}_{i-1}$, and $\upvar{\vec{a}}_{i}$ and
    $\upvar{\vec{a}}_{i}$ are elements of the lattice
    $(\upvar{\widetilde{\terms}}_n, \rel)$ such that $\upvar{\vec{a}}_{i-1}
    \rel \upvar{\vec{a}}_{i}$.  Therefore, \cref{alg:cspkpn2}, which
    iteratively calls $\af_\con$, terminates after a finite number of steps if
    the lattices have a finite height.

    The semilattice for symbols has a fixed height of two element (the $\nil$
    and the symbol itself).  The rest of the terms, which represent
    collections, may ``expand'' only a finite number of times for a given
    $\con$ (by expansion we mean adding new elements to a collection, which
    leads to term coercion in the semilattice).  The size of a tuple is fixed.
    A record and a choice are expanded by adding elements with labels that are
    not yet presented in the collection.  The set of labels in $\con$ is finite
    and the algorithm cannot generate new labels.  Therefore, the record and
    the choice can expand only a finite number of times.  Therefore, the
    lattices have a finite height.

\end{proof}


Substitution of variables with ground terms is a monotonic function.  Below we
prove that the substitution of down-coerced variables is a decreasing function.
Similarly, we can prove that the substitution of up-coerced variables is an
increasing function in the same way.
\begin{proposition}[Substitution monotonicity]
    \label{prop:monotonicity}
    Let $t$ be a term such that $|\bvars(t)| = \emptyset$, $\downvar{\vec{v}} = (v_1,\dots,v_k)$ be a vector of down-coerced variables in $t$, and $\downvar{\vec{s_1}} = (s^1_1,\dots,s^1_k)$ and
    $\downvar{\vec{s_2}} = (s^2_1,\dots, s^2_k)$ be vectors of down-coerced
    ground terms such that $\downvar{\vec{s_1}} \rel \downvar{\vec{s_2}}$.  Then
    \[
        t[\downvar{\vec{v}}/\downvar{\vec{s_1}}] \rel
        t[\downvar{\vec{v}}/\downvar{\vec{s_2}}]\text{.}
    \]
\end{proposition}
\begin{proof}
    The monotonicity of the substitution follows from the structure of the
    seniority relation.  Any term is covariant with respect to its subterms
    (see~\cref{def:seniority}).

\end{proof}


The function $\af$ produces the ``tightest'' approximation.  Furthermore, the
tightest approximation is unique.
\begin{lemma}
    \label{lem:tightest}
    Assume a constraint $t_1 \rel t_2$ and approximations
    $(\downvar{\vec{a_1}}, \upvar{\vec{a_1}})$ and $(\downvar{\vec{a_2}},
    \upvar{\vec{a_2}})$ such that $\af(t_1 \rel t_2, \downvar{\vec{a_1}},
    \upvar{\vec{a_1}}) = (\downvar{\vec{a_2}}, \upvar{\vec{a_2}})$ are given.
    If
    \[
        t_1[\downvar{\vec{v}}/\downvar{\vec{a}}_2,\upvar{\vec{v}}/\upvar{\vec{a}}_1] \rel
        t_2[\downvar{\vec{v}}/\downvar{\vec{a}}_1,\upvar{\vec{v}}/\upvar{\vec{a}}_2]\text{,}
    \]
    then:
    \begin{enumerate}
        \item No approximation $(\downvar{\vec{a_3}}, \upvar{\vec{a_3}})$
            exists such that $(\downvar{\vec{a_3}}, \upvar{\vec{a_3}}) \neq
            (\downvar{\vec{a_2}}, \upvar{\vec{a_2}})$, $\downvar{\vec{a_2}}
            \rel \downvar{\vec{a_3}}$, $\upvar{\vec{a_3}} \rel
            \upvar{\vec{a_2}}$ and
            \begin{equation}
                \label{eq:tightest2}
                t_1[\downvar{\vec{v}}/\downvar{\vec{a}}_2,\upvar{\vec{v}}/\upvar{\vec{a}}_1] \rel
                t_1[\downvar{\vec{v}}/\downvar{\vec{a}}_3,\upvar{\vec{v}}/\upvar{\vec{a}}_1] \rel
                t_2[\downvar{\vec{v}}/\downvar{\vec{a}}_1,\upvar{\vec{v}}/\upvar{\vec{a}}_3]\text{.}
            \end{equation}
        \item For any other approximation $(\downvar{\vec{a'_2}},
            \upvar{\vec{a'_2}})$ such that
            \begin{equation}
                \label{eq:tightest3}
                t_1[\downvar{\vec{v}}/\downvar{\vec{a'}}_2,\upvar{\vec{v}}/\upvar{\vec{a}}_1] \rel
                t_2[\downvar{\vec{v}}/\downvar{\vec{a}}_1,\upvar{\vec{v}}/\upvar{\vec{a'}}_2]
            \end{equation}
            there exists $(\downvar{\vec{a'_3}}, \upvar{\vec{a'_3}})$ such that
            \[
                t_1[\downvar{\vec{v}}/\downvar{\vec{a'}}_2,\upvar{\vec{v}}/\upvar{\vec{a}}_1] \rel
                t_1[\downvar{\vec{v}}/\downvar{\vec{a'}}_3,\upvar{\vec{v}}/\upvar{\vec{a}}_1] \rel
                t_2[\downvar{\vec{v}}/\downvar{\vec{a}}_1,\upvar{\vec{v}}/\upvar{\vec{a'}}_3]\text{.}
            \]
    \end{enumerate}
\end{lemma}
\begin{proof}
    \begin{enumerate}
        \item By the definition in \cref{sec:approximating-function} the
            function $\af$ coerces the approximation $(\downvar{\vec{a_1}},
            \upvar{\vec{a_1}})$ only if the coercion is required satisfaction
            of $t_1 \rel t_2$ The function produces $(\downvar{\vec{a_2}},
            \upvar{\vec{a_2}})$ as a result.  The approximation
            $(\downvar{\vec{a_3}}, \upvar{\vec{a_3}})$ such that
            (\ref{eq:tightest2}) holds could only exist if $\af$ performed
            excessive coercions, which always can be avoided.
        \item The uniqueness of $(\downvar{\vec{a_2}}, \upvar{\vec{a_2}})$
            follows from the definition of the seniority relation
            (\cref{def:seniority}).
    \end{enumerate}
\end{proof}

\begin{lemma}
    \label{lem:max-fixed-point}
    Assume a set of constraints $\con$, $\bvars(\con) = \emptyset$, is given.
    Let for $k > 0$
    \[
        (\downvar{\vec{a}}_0, \upvar{\vec{a}}_0), \dots,(\downvar{\vec{a}}_k, \upvar{\vec{a}}_k)
    \]
    be a chain of approximations such that $(\downvar{\vec{a}}_i,
    \upvar{\vec{a}}_i) = \af_\con(\downvar{\vec{a}}_{i-1},
    \upvar{\vec{a}}_{i-1})$ for any $0 < i \leq k$, and $\downvar{\vec{a}}_0
    = (\nil, \dots, \nil)$ and $\upvar{\vec{a}}_0 = (\none, \dots, \none)$.
    Then for any fixed-point $(\downvar{\vec{a}}, \upvar{\vec{a}})$
    \begin{equation}
        \label{eq:best-solution}
        \downvar{\vec{a}} \rel \downvar{\vec{a}_k} \text{ and } \upvar{\vec{a}_k} \rel \upvar{\vec{a}}\text{.}
    \end{equation}
\end{lemma}
\begin{proof}
    The proof consists of two parts.  First, we prove that a fixed-point
    $(\downvar{\vec{s}}, \upvar{\vec{s}})$ with property
    (\ref{eq:best-solution}) exists.  Then we show that the $\af_\con$
    converges to $(\downvar{\vec{s}}, \upvar{\vec{s}})$, i.e.
    $(\downvar{\vec{a}}_k, \upvar{\vec{a}}_k) = (\downvar{\vec{s}},
    \upvar{\vec{s}})$.
    \begin{description}
        \item[Existence] $(\downvar{\widetilde{\terms}}_m, \rel)$ and
            $(\upvar{\widetilde{\terms}}_n, \rel)$ are complete lattices and
            $\af_\con$ is an order-preserving function.  By Knaster-Tarski
            theorem~\cite{tarski1955lattice}, the sets of fixed points of
            $\af_\con$ in $(\downvar{\widetilde{\terms}}_m, \rel)$ and
            $(\upvar{\widetilde{\terms}}_n, \rel)$ are complete lattices too.
            Therefore, there exists the fixed-point $(\downvar{\vec{s}},
            \upvar{\vec{s}})$ such that for any fixed-point
            $(\downvar{\vec{\overline{s}}}, \upvar{\vec{\overline{s}}})$,
            $\downvar{\vec{\overline{s}}} \rel \downvar{\vec{s}}$ and
            $\downvar{\vec{s}} \rel \downvar{\vec{\overline{s}}}$.
    \item[Reachability]
            Proof by contradiction.  Assume that $\af_\con$ does not converge
            to $(\downvar{\vec{s}}, \upvar{\vec{s}})$, i.e.
            $(\downvar{\vec{a}}_k, \upvar{\vec{a}}_k)
            = (\downvar{\vec{\overline{s}}}, \upvar{\vec{\overline{s}}})$,
            where $(\downvar{\vec{\overline{s}}}, \upvar{\vec{\overline{s}}})
            \neq (\downvar{\vec{s}}, \upvar{\vec{s}})$, and
            $\downvar{\vec{\overline{s}}} \rel \downvar{\vec{s}}$ or
            $\downvar{\vec{s}} \rel \downvar{\vec{\overline{s}}}$.  Let
            $\downvar{\vec{\overline{s}}} \rel \downvar{\vec{s}}$ (the case
            when $\downvar{\vec{s}} \rel \downvar{\vec{\overline{s}}}$ is
            considered similarly).

            Let $(\downvar{\vec{a}}_{i-1}, \upvar{\vec{a}}_{i-1})$ be the
            approximation that precedes $(\downvar{\vec{\overline{s}}},
            \upvar{\vec{\overline{s}}})$ in the chain of approximations:
            $\af_\con(\downvar{\vec{a}}_{i-1}, \upvar{\vec{a}}_{i-1})
            = (\downvar{\vec{\overline{s}}}, \upvar{\vec{\overline{s}}})$.
            For every constraint $t_1 \rel t_2 \in \con$
            \[
                t_1[\downvar{\vec{v}}/\downvar{\vec{\overline{s}}},\upvar{\vec{v}}/\upvar{\vec{\overline{a}}}_{i-1}] \rel
                t_2[\downvar{\vec{v}}/\downvar{\vec{a}}_{i-1},\upvar{\vec{v}}/\upvar{\vec{\overline{s}}}]\text{.}
            \]
            Since $\downvar{\vec{s}}$ is a fixed point, then
            \[
                t_1[\downvar{\vec{v}}/\downvar{\vec{s}},\upvar{\vec{v}}/\upvar{\vec{a}}_{i-1}] \rel
                t_2[\downvar{\vec{v}}/\downvar{\vec{a}}_{i-1},\upvar{\vec{v}}/\upvar{\vec{s}}]\text{.}
            \]
            On the other hand, $\downvar{\vec{\overline{s}}} \rel
            \downvar{\vec{s}}$. Due to substitution monotonicity
            (\cref{prop:monotonicity}),
            \begin{equation}
                \label{eq:contradiction}
                t_1[\downvar{\vec{v}}/\downvar{\vec{\overline{s}}},\upvar{\vec{v}}/\upvar{\vec{a}}_{i-1}] \rel t_1[\downvar{\vec{v}}/\downvar{\vec{s}},\upvar{\vec{v}}/\upvar{\vec{a}}_{i-1}]\text{.}
            \end{equation}
            It contradicts \cref{lem:tightest}, which states that $\af_\con$
            produces the ``tightest'' approximation
            $\downvar{\vec{\overline{s}}}$, but according to
            (\ref{eq:contradiction}) it follows that $\downvar{\vec{s}}$ is the
            ``tightest''.

            Therefore, $\af_\con$ converges to the fixed point
            $\downvar{\vec{s}}$ and no $\downvar{\vec{\overline{s}}}$ exists
            such that $\downvar{\vec{s}} \rel \downvar{\vec{\overline{s}}}$.
    \end{description}
\end{proof}

\begin{theorem}[Correctness]
    \label{thm:cspkpn2-correct}
    For any set of constraints $\con$ such that $\bvars(\con) = \emptyset$,
    \cspws{} for $\con$ is unsatisfiable iff \cref{alg:cspkpn2} returns $\unsat$.
\end{theorem}
\begin{proof}
    Proof by contradiction.

    ($\Rightarrow$) Let $\con$ be an unsatisfiable set of constraints and
    \cref{alg:cspkpn2} returns $(\downvar{\vec{s}}, \upvar{\vec{s}})$ such that
    $(\downvar{\vec{s}}, \upvar{\vec{s}}) \neq (\tbot, \ttop)$.
    $(\downvar{\vec{s}}, \upvar{\vec{s}})$ is the fixed point that contains
    values satisfying $\con$.  This contradicts the initial hypothesis.
    Therefore, \cref{alg:cspkpn2} returns $\unsat$ if $\con$ is unsatisfiable.

    ($\Leftarrow$) Let \cref{alg:cspkpn2} return $\unsat$ and $\con$ has
    a solution.  In this case the chain of approximations in \cref{alg:cspkpn2}
    returns $(\tbot, \ttop)$.  This is the fixed point and by
    \cref{lem:max-fixed-point} no other fixed point
    $(\downvar{\vec{\overline{s}}}, \upvar{\vec{\overline{s}}})$ exists such
    that $\tbot \rel \downvar{\vec{\overline{s}}}$ or
    $\upvar{\vec{\overline{s}}} \rel \ttop$, which means that no fixed points
    apart from $(\tbot, \ttop)$ exists.  This contradicts the initial hypothesis.
    Therefore, $\con$ is unsatisfiable if \cref{alg:cspkpn2} returns $\unsat$.
\end{proof}

\section{\cspws{} Algorithm}
\label{sec:algorithm}

A straightforward algorithm for $\cspws{}$ has to run  \cref{alg:cspkpn2} for
each of $2^l$ pairs of the semi-lattices, where $l=|\bvars(\con)|$.  Instead,
we present   iterative \cref{alg:cspkpn} which takes the advantage   of the
order-theoretical structure of the MDL and  generates an adjunct SAT problem on
the way.


\begin{algorithm}[htb]
    \caption{$\cspws(\con)$}
\label{alg:cspkpn}
\begin{algorithmic}[1]
    \State$c \gets |\con|$
    \State$i \gets 0$
    \State$\ssat_{0} \gets \emptyset$
    \State$\downvar{\vec{a}}_{0} \gets (\nil,\dots,\nil)$
    \State$\upvar{\vec{a}}_{0} \gets (\none,\dots,\none)$
    \Repeat \label{alg:cspkpn-loop-start}
       \State$i \gets i + 1$
    	\State$(\downvar{\vec{a}}_i, \upvar{\vec{a}}_i) \gets \af_\con(\downvar{\vec{a}}_{i-1}, \upvar{\vec{a}}_{i-1})$
        \State$\ssat_i \gets \ssat_{i-1} \cup \bigcup\limits_{t_1 \rel t_2 \in \con}  (\wf(t_1[\vec{v}/\vec{a}_i]) \cup \wf(t_2[\vec{v}/\vec{a}_i]) \cup \sen(t_1[\vec{v}/\vec{a}_i] \rel t_2[\vec{v}/\vec{a}_i]))$
    \Until{$(\mathsf{SAT}(\ssat_{i}), \downvar{\vec{a}}_i, \upvar{\vec{a}}_i) = (\mathsf{SAT}(\ssat_{i-1}), \downvar{\vec{a}}_{i-1}, \upvar{\vec{a}}_{i-1})$}
    \If{$\ssat_{i}$ is unsatisfiable}
        \State \Return$\unsat$
    \Else
        \State \Return$(\vec{b}, \downvar{\vec{a}}_i[\vec{f}/\vec{b}], \upvar{\vec{a}}_i[\vec{f}/\vec{b}])$, where $\vec{b} \in \mathsf{SAT}(\ssat_i)$
    \EndIf
\end{algorithmic}
\end{algorithm}

Let $\ssat_0\subseteq \ssat_1\subseteq \dots \subseteq \ssat_s$ be sets of
Boolean constraints, and $\downvar{\vec{a}}$ and $\upvar{\vec{a}}$ be vectors
of semiground terms such that $|\downvar{\vec{a}}| = |\dtvars(\con)|$ and
$|\upvar{\vec{a}}| = |\utvars(\con)|$. 
We seek the solution as a fixed point of a chain of approximations in the
following form:
\[
    (\ssat_0, \downvar{\vec{a}}_0, \upvar{\vec{a}}_0), \ldots,
    (\ssat_{s-1}, \downvar{\vec{a}}_{s-1}, \upvar{\vec{a}}_{s-1}),
    (\ssat_s, \downvar{\vec{a}}_s, \upvar{\vec{a}}_s)\text{,}
\]
where for every $1\leq i \le s$
and a vector of Boolean values $\vec{b}$ that is a solution to
$\mathsf{SAT}(\ssat_i)$:
\[
    \downvar{\vec{a}}_{i}[\vec{f}/\vec{b}] \rel \downvar{\vec{a}}_{i-1}[\vec{f}/\vec{b}]
    \qquad\text{and}\qquad
    \upvar{\vec{a}}_{i-1}[\vec{f}/\vec{b}] \rel \upvar{\vec{a}}_{i}[\vec{f}/\vec{b}]\text{.}
\]


The adjunct set of Boolean constraints potentially expands at every iteration
of the algorithm by inclusion of further logic formulas produced by the set of Boolean constraint  $\wf$  (see \cref{fig:wf-constraints}) ensuring  well-formdness of the terms  and
the set of Boolean constraints  $\sen$ (see  \cref{fig:b-sen-con}) ensuring that the seniority relations holds.  
The starting point is $\ssat_0=\emptyset$,
$\downvar{\vec{a}}_0=(\nil,\dots,\nil)$,
$\upvar{\vec{a}}_0=(\none,\dots,\none)$ and the chain terminates as soon as
$\mathsf{SAT}(\ssat_s) = \mathsf{SAT}(\ssat_{s-1})$, $\upvar{\vec{a}}_{s}
= \upvar{\vec{a}}_{s-1}$, $\downvar{\vec{a}}_{s} = \downvar{\vec{a}}_{s-1}$, where by $\mathsf{SAT}(\ssat_i)$ we mean
a set of Boolean vector satisfying $\ssat_i$. 
Whether the set of Boolean constraints actually expands or not can be
determined by checking the satisfiability of $\mathsf{SAT}(\ssat_i)\neq
\mathsf{SAT}(\ssat_{i-1})$ for the current iteration $i$.


\begin{figure}[t]
\fbox{\begin{minipage}{\textwidth}
    \normalfont
    \begin{enumerate}
        \item $\wf(t) = \emptyset$ if $t$ is a symbol;
        \item $\wf(t) = \bigcup_{1 \leq i \leq n} \wf(t_i)$ if $t$ is
            a tuple $\tuple{t_1 \dots t_n}$;
        \item $\wf(t) = \{\lnot (g_i \land g_j)~|~1 \leq i \neq j \leq n\text{
            and } l_i = l_j\} \cup \bigcup_{1 \leq i \leq n} \{g_i \to g~|~g
            \in \wf(t_i)\}$ if $t$ is a record
            $\record{\el{l_1}{g_1}{t_1},\dots,\el{l_n}{g_n}{t_n}}{}$ or
            a choice $\choice{\el{l_1}{g_1}{t_1},\dots,\el{l_n}{g_n}{t_n}}{}$;
        \item $\wf(t) = \{\lnot (g_i \land g_j)~|~1 \leq i \neq j \leq n\} \cup
            \{\bigvee_{1 \leq i \leq n} g_i\} \cup \bigcup_{1 \leq
            i \leq n} \{g_i \to g~|~g \in \wf(t_i)\}$ if $t$ is a switch
            $\switch{\el{}{g_1}{t_1},\dots,\el{}{g_n}{t_n}}$.
    \end{enumerate}
\end{minipage}}
\caption{The set of Boolean constraints that ensures well-formedness of a  term $t$}
\label{fig:wf-constraints}
\end{figure}

\begin{figure}[t]
\fbox{\begin{minipage}{\textwidth}
    \normalfont
    \label{defn:b-sen-con}
    \begin{enumerate}
        \item $\sen(t_1 \rel t_2) = \emptyset$, if $t_1$ and $t_2$ are equal symbols.
        \item $\sen(t_1 \rel t_2) = \bigcup_{1 \leq i \leq k} \sen(t^1_i \rel
            t^2_i)$, if $t_1$ is a tuple $\tuple{t^1_1\dots t^1_k}$ and $t_2$
            is a tuple $\tuple{t^2_1\dots t^2_k}$;
        \item $\sen(t_1 \rel t_2) = \bigcup_{1 \leq j \leq m} \sen_j(t^2_j)$,
            if $t_1$ is a record
            $\record{\el{l^1_1}{g^1_1}{t^1_1},\dots,\el{l^1_k}{g^1_k}{t^1_k}}{}$,
            $t_2$ is a record
            $\record{\el{l^1_2}{g^2_1}{t^2_1},\dots,\el{l^2_m}{g^2_m}{t^2_m}}{}$
            and $\sen_j(t^2_j)$ is one of the following:
            \begin{enumerate}
                \item $\sen_j(t^2_j) = \{(g^1_i \land g^2_j) \to g~|~g \in
                    \sen(t^1_i \rel t^2_j)\}$, if $\exists i: 1 \leq i \leq k$
                    and $l^1_i = l^2_j$;
                \item $\sen_j(t^2_j) = \{\lnot g^2_j\}$, otherwise;
            \end{enumerate}
        \item $\sen(t_1 \rel t_2) = \bigcup_{1 \leq i \leq m} \sen_i(t^1_i)$,
            if $t_1$ is a choice
            $\choice{\el{l^1_1}{g^1_1}{t^1_1},\dots,\el{l^1_k}{g^1_k}{t^1_k}}{}$,
            $t_2$ is a choice
            $\choice{\el{l^2_1}{g^2_1}{t^2_1},\dots,\el{l^2_m}{g^2_m}{t^2_m}}{}$
            and $\sen_i(t^1_i)$ is one of the following:
            \begin{enumerate}
                \item $\sen_i(t^1_i) = \{(g^1_i \land g^2_j) \to g~|~g \in
                    \sen(t^1_i \rel t^2_j)\}$, if $\exists j: 1 \leq i \leq m$
                    and $l^1_i = l^2_j$;
                \item $\sen_i(t^1_i) = \{\lnot g^1_i\}$, otherwise;
            \end{enumerate}
        \item $\sen(t_1 \rel t_2) = \{g^1_i \to g~|~1 \leq i \leq k\text{ and
            } g \in \sen(t^1_i \rel t^2_i)\}$, if $t_1$ is a switch
            $\switch{\el{}{g^1_1}{t^1_1},\dots\el{}{g^1_k}{t^1_k}}$ and $t_2$
            is an arbitrary term.
        \item $\sen(t_ 1\rel t_2) = \{g^2_i \to g~|~1 \leq i \leq k\text{ and
            } g \in \sen(t_1 \rel t^2_i)\}$, if $t_1$ is an arbitrary term and
            $t_2$ is a switch
            $\switch{\el{}{g^2_1}{t^2_1},\dots\el{}{g^2_k}{t^2_k}}$.
        \item $\sen(t_1 \rel t_2) = \{\false\}$, otherwise.
    \end{enumerate}
\end{minipage}}
\caption{The set of Boolean constraints that ensures the seniority relation  $t_1 \rel t_2$}
\label{fig:b-sen-con}
\end{figure}

We argue that if the original \cspws{} is satisfiable then so is
$\mathsf{SAT}(\ssat_s)$ and that the tuple of vectors $(\vec{b}_s,
\downvar{\vec{a}}_s[\vec{f}/\vec{b_s}], \upvar{\vec{a}}_s[\vec{f}/\vec{b_s}])$
is a solution to the former, where $\vec{b_s}$ is a solution of
$\mathsf{SAT}(\ssat_s)$.  In other words, the iterations terminate when the
conditional approximation limits the t-variables, and when the adjunct SAT
constrains the b-variables enough to ensure the satisfaction of all \cspws{}
constraints.  In general, the set $\mathsf{SAT}(\ssat_s)$ can have more than
one solution and we select one of them.  Heuristics that allows to choose
a solution that is better for the given application is left for further
research.

{\bf Implementation.} We implemented the $\cspws{}$ algorithm as a solver in
the OCaml language. The input for the solver is a set of constraints and the
output is in the form of assignments to b-variables and t-variables.  It works
on top of the PicoSAT~\cite{picosat} library (although any other SAT solver
could be used instead). PicoSAT is employed as a subsolver that deals with
Boolean assertions.

\section{Conclusion and Future Work}
\label{sec:conclusion}

We have presented a new mechanism for choreographing service interfaces based
on CSP and SAT that configures generic non-local interfaces in the context.  We
developed a Message Definition Language that can be used in the context of
service-based applications.  Our mechanism supports subtyping, polymorphism and
inheritance, thanks to the order relation defined on MDL terms.  We presented
the CSP solution algorithm for interface configuration, which has been
developed specifically for this problem.

In the context of Cloud, our results may prove useful to the
software-as-service community since we can support much more generic interfaces
than are currently available.  Building services the way we do could enable
service providers to configure a solution for a network customer based on
services that they have at their disposal as well as those provided by other
providers and the customer themselves, all solely on the basis of interface
definitions and automatic tuning to non-local requirements.

The next step will be the design of a mechanism for automatic interface
derivation from code of the services, which can be done in a straightforward
manner.  This brings an advantage over choreography mechanisms that rely on
behavioural protocols: automatic derivation of the behaviour from the code is
a difficult problem that have not been solved yet.

\bibliographystyle{splncs03}
\bibliography{paper}

\end{document}